\documentclass{llncs}
\usepackage{times}
\usepackage{amssymb}
\usepackage{url}
\usepackage{epsfig}
\usepackage{ifthen}
\usepackage{cite}
\usepackage{color}

\newcommand{\eps}{\ensuremath{\varepsilon}}
\begin{document}

\mainmatter

\title{Detecting Useless Transitions in Pushdown Automata}

\author{Wan Fokkink \and Dick Grune \and Brinio Hond \and Peter Rutgers}
\institute{VU University Amsterdam, Department of Computer Science}

\maketitle

\begin{abstract}
Pushdown automata may contain transitions that are never used in any accepting run of the automaton.
We present an algorithm for detecting such useless transitions.
A finite automaton that captures the possible stack content during runs of the pushdown automaton,
is first constructed in a forward procedure to determine which transitions are reachable, and then
employed in a backward procedure to determine which of these transitions can lead to a final state.
\end{abstract}

\section{Introduction}\label{sec:intro2}

Context-free languages are used in language specification, parsing, and code optimization.
They are defined by means of a context-free grammar or a pushdown automaton (pda).
Some languages can be specified more efficiently by a pda
than by a context-free grammar, as shown by Goldstine, Price, and Wotschke \cite{GoldstinePW82econ}.
Pda's are at the root of deterministic parsers for context-free languages (notably LL, LR),
see e.g.\ \cite{Substring.Nederhof.1996,LR.Bertsch.1999}.
We consider pda's in which any number of symbols can be popped
from as well as pushed onto the stack in one transition.
%and do not impose restrictions on the initial stack.
Popping zero or multiple symbols is useful in bottom-up parsing, and facilitates the reversal of a pda.
%, while pushing multiple symbols is useful in top-down parsing.
%Allowing arbitrary initial stacks is useful for parsing as well, or when the pda is
%embedded as a sub-automaton in another pda.

For context-free grammars, it is rather straightforward to determine
whether a production is \emph{useless}, i.e., cannot occur in a derivation from
the start variable to a string of terminal symbols; such a method is discussed
in many textbooks on formal languages (e.g., \cite[Theorem 6.2]{Linz11}).
It consists of two parts: Detect which variables are reachable from the start variable,
and which variables can be transformed into a string of terminal symbols.
Productions that contain a useless variable, not satisfying these two properties,
can be removed from the grammar without changing the associated language.

We present an algorithm to detect whether a transition in a pda is useless,
meaning that no run of the pda from the initial configuration to a final state
includes this transition. Such a transition can be removed from the pda without
changing the language accepted by the pda. Removing useless transitions, which
may improve the performance of running the pda, is especially sensible if
the pda has been generated automatically, because then there tend to be useless
transitions present.

Similar to detecting useless variables in context-free grammars,
our algorithm for detecting useless transitions in a pda consists of
two parts. The first part finds which transitions are not reachable from
the initial configuration.
Here we exploit an algorithm by Finkel, Willems and Wolper~\cite{FinkelWW97}
to construct a finite automaton (nfa) that captures exactly all
possible stacks in the reachable configurations of a pda.
Their approach is modified to take into account that multiple symbols may be popped from
the stack at once. The second part of our algorithm finds after which transitions
it is impossible to reach a final state.  Here we use the nfa
constructed in the first part to compute in a backward fashion which transitions
can lead to a final state in the pda.

We prove that the algorithm marks exactly the useless transitions.
The worst-case time complexity of the algorithm is $O(Q^4T)$, with
$Q$ the number of states and $T$ the number of transitions of the pda.
This worst case actually only occurs in the unlikely case that the nfa
is constructed over a large number of iterations, is saturated with
$\eps$-transitions, and contains a lot of backward nondeterminism.
A prototype implementation of the algorithm exhibits a good performance.
%% , for example on a grammar of C transformed into a pda.

An alternative approach is to transform the pda under consideration into an equivalent
context-free grammar, and then determine the useless productions.
Another alternative approach is to check for each transition separately whether it is useless:
Provide the transition with a special input symbol $\xi$, all other transitions in the pda
with empty input $\eps$, and check whether the language accepted by the resulting pda intersected
with the regular language $\xi^+$ is empty. Checking emptiness of pda's is generally
performed by a conversion to a context-free grammar.
Disadvantage of these approaches is that the resulting grammar tends to be much larger than the original pda.

Bouajjani, Esparza and Maler \cite{BouajjaniEM97} employed a method similar to the one in \cite{FinkelWW97}
to capture the reachable configurations of a pda via an nfa, in the context of model checking infinite-state systems.
Griffin~\cite{Griffin06} showed how to detect which transitions are reachable
from the initial configuration in a deterministic pushdown automaton (dpda).
For each transition, the algorithm creates a temporary dpda in which the
successive state of the transition is set to a new, final state; all other
states in the temporary dpda are made non-final. Then it is checked whether the
language generated by the dpda is empty; if it is, the transition is unreachable.
This algorithm determines which transitions are reachable
from the initial configuration, but not which transitions can lead to a final state.
Vice versa, Kutrib and Malcher \cite{KutribMalcher12} studied reversibility of dpda's.

%% Goldstine, Price and Wotschke developed algorithms
%% to optimize a pda. If for an application it is preferable to have fewer states and more stack symbols,
%% the number of states can be reduced at the price of extra stack symbols \cite{GoldstinePW82}.
%% On the other hand, if it is better to have few stack symbols, they can be exchanged
%% for extra states \cite{GoldstinePW93}.

\section{Preliminaries}

\begin{definition}\label{def:NPDA}
A nondeterministic pushdown automaton (pda) consists of
a finite set of states $Q$, a finite input alphabet $\Sigma$, a finite stack alphabet $\Gamma$,
a finite transition relation $\delta:Q \times (\Sigma\cup\{\eps\}) \times \Gamma^* \rightarrow 2^{Q \times \Gamma^*}$,
an initial state $q_0$, and a set $F$ of final states.
\end{definition}

\noindent
In this definition, $\eps$ denotes the empty string. Note that
zero or multiple symbols can be popped from the stack in one transition.
It is assumed that the initial stack is empty. (An arbitrary initial stack $\sigma$ can be constructed
by adding a new initial state $\hat{q_0}$ and a transition $\delta(\hat{q_0},\eps,\eps)=\{(q_0,\sigma)\}$.)

A configuration consists of a state from $Q$ together with a stack from $\Gamma^*$.
We let $a,b,c,d$ denote elements in $\Gamma$, and $\rho,\sigma,\tau,\upsilon,\zeta$ strings in $\Gamma^*$.
The reverse of a string $\sigma$ is denoted by $\sigma^R$.
A transition $(r,\tau)\in\delta(q,a,\sigma)$ (or $(r,\tau)\in\delta(q,\eps,\sigma)$)
gives rise to moves $(q,\sigma\rho)\stackrel{a}{\rightarrow}(r,\tau\rho)$ (or $(q,\sigma\rho)\stackrel{\eps}{\rightarrow}(r,\tau\rho)$)
between configurations, for any $\rho$.
The language accepted by a pda consists of the strings in $\Sigma^*$ that give rise to a run
of the pda from the initial configuration $(q_0,\eps)$ to a configuration $(r,\sigma)$ with $r\in F$.

A transition in a pda is \emph{useless} if no run of the pda from the initial configuration
$(q_0,\eps)$ to a configuration $(r,\sigma)$ with $r\in F$, for any input string from $\Sigma^*$,
includes this transition. To determine the useless transitions, input strings from $\Sigma^*$ are irrelevant.
The point is that, since a run for any input string suffices to make a transition useful,
we can assume that any desired terminal symbol from $\Sigma$ is available as input at any time.
Input strings from $\Sigma$ are therefore disregarded in our algorithm to detect useless transitions.
(In the context of model-checking infinite-state systems, pda's in which input strings are disregarded
are called ``pushdown systems'' or ``pushdown processes'' \cite{Walukiewicz96}.)

A transition $(r,\tau)\in\delta(q,\sigma)$ is written as $q\stackrel{\sigma/\tau}{\rightarrow}r$.
It gives rise to moves $(q,\sigma\rho)\rightarrow(r,\tau\rho)$.
We write $(s,\upsilon)\rightarrow^*(t,\zeta)$ if there is a run from $(s,\upsilon)$ to $(t,\zeta)$ of the pda,
consisting of zero or more moves.

\begin{definition}\label{def:NFA}
A nondeterministic finite automaton (nfa) consists of
a finite set of states $Q$, a finite input alphabet $\Sigma$,
a transition relation $\delta:Q \times (\Sigma\cup\{\eps\}) \rightarrow 2^Q$, an initial state $q_0$,
and a set $F$ of final states.
\end{definition}

\noindent
In our application of nfa's, the input alphabet is the stack alphabet $\Gamma$ from the pda.

A transition $r\in\delta(q,a)$ is written as $q\stackrel{a}{\rightarrow}r$.
We write $q\stackrel{a_1...a_k}{\leadsto}r$ if there is a path from $q$ to $r$ in the nfa
with consecutive labels $a_1,...,a_k\in\Gamma$.
We write $q\stackrel{a_1...a_k}{\Longrightarrow}r$ if there is such a path from $q$ to $r$,
possibly intertwined with transitions labeled by $\eps$.
The language accepted by an nfa consists of the strings $a_1...a_k$ in $\Sigma^*$ for which there
exists a path $q_0\stackrel{a_1...a_k}{\Longrightarrow}r$ with $r\in F$.

\section{Detecting the useless transitions in a pda}

Our algorithm for detecting useless transitions in a pda summarizes all
reachable configurations of the pda in an nfa. As a first step, an nfa is constructed
that accepts the stacks that can occur during any run of the pda.
A second step determines which transitions can lead to a configuration
from which a final state can be reached. Transitions that cannot be reached from the
initial start (as determined in step $1$), or that cannot lead to a final state (as determined in step $2$)
are useless.

\subsection{Detecting the unreachable transitions}
\label{sec:forward}

A configuration or transition in a pda $P$ is reachable if it is employed in a run of $P$,
starting from the initial configuration. The reachable configurations of $P$ are captured by means of an nfa $N$.
The stacks in $\Gamma^*$ that can occur at a state $q$ in $P$ are accepted at the state $q$ in $N$,
in reverse order. During the construction of $N$, intermediate non-final states
are created when multiple symbols are pushed onto the stack in one transition.
They are denoted by $n,m$, to distinguish them from the final states $q,r,s,t$ that are
inherited from $P$. A state in $N$ that may be either final or non-final is denoted by $x,y,z$.

Fix a pda $P=(Q, \Sigma, \Gamma, \delta, q_0, F)$; as said, we will disregard $\Sigma$.
To achieve a single final state without outgoing transitions that is only reached with an empty stack, a fresh stack symbol $b_0$
is added to $\Gamma$, and the initial stack is $b_0$ (instead of $\eps$). In each run of the pda, $b_0$ is always at the bottom of the stack.
Fresh states $q_e$ and $q_f$ are added to $Q$, and $\delta$ is extended with transitions $q\stackrel{\eps/\eps}{\rightarrow}q_e$ for every $q\in F$,
$q_e\stackrel{a/\eps}{\rightarrow}q_e$ for every $a\in\Gamma\backslash\{b_0\}$, and $q_e\stackrel{b_0/\eps}{\rightarrow}q_f$.
We change $F$ to $\{q_f\}$. The resulting pda is called $P_0$.

Initially the nfa $N$ under construction consists of the transition $m_0\stackrel{b_0}{\rightarrow}q_0$;
the fresh state $m_0$ is non-final and $q_0$ final. Intuitively, this transition builds the initial stack of $P_0$.
The set $U_1$ of unreachable transitions in $P_0$ initially, as an overapproximation, contains all transitions in $P_0$.
The nfa $N$ and the set $U_1$ are constructed as follows.

\vspace{1.5mm}

\noindent
Procedure {\em forward}: For each transition $\theta=q\stackrel{\sigma/\tau}{\rightarrow}r$ in $P_0$ do:
\begin{enumerate}
\item[1.~~]
If $q$ is not a state in $N$, then stop this iteration step.\vspace{1mm}
\item[2.~~]
Determine the set $S_{q,\sigma}$, which either consists only of $q$, if $\sigma=\eps$, or of the states $n$ for which
there exists a path $n\stackrel{a}{\rightarrow}y\stackrel{\sigma'^R}{\Longrightarrow}q$ in $N$, if $\sigma=\sigma'a$.\vspace{1mm}
\item[3.~~]
If $S_{q,\sigma}=\emptyset$, then stop this iteration step.
\item[4.~~]
If $\theta\in U_1$, then delete $\theta$ from $U_1$, and establish a path $\stackrel{\tau^R}{\leadsto}r$ in $N$ (see below); the state $r$ in $N$ is final.\vspace{1mm}
\item[5.~~]
Let $y$ be the first state in the path $\stackrel{\tau^R}{\leadsto}r$. For each state $x\in S_{q,\sigma}$, if the transition $x\stackrel{\eps}{\rightarrow}y$
is not yet present in $N$, then add this transition to $N$.
\end{enumerate}

\vspace{1mm}

\noindent
If $N$ changed during this run, then perform the {\em forward} procedure again, over all transitions in $P_0$.
Else, stop, and return the constructed nfa $N$ and the set $U_1$ of unreachable transitions in $P_0$.
These transitions are then culled from $P_0$, producing the pda $P_1$.
The sets $S_{q,\sigma}$ need to be recomputed in every run of the {\em forward} procedure.
The sets $S_{q,\sigma}$ computed in the last run of the {\em forward} procedure are stored,
as they will be used in the {\em backward} procedure in the next section.

The procedure called in step 4, which establishes a path $\stackrel{a_1...a_k}{\leadsto}z$ in $N$
and returns the first state in this path, is defined as follows.
\begin{itemize}
\item[4.1]
If $k=0$, then return $z$, and stop.\vspace{1mm}
\item[4.2]
If there is a transition $n\stackrel{a_k}{\rightarrow}z$ in $N$ with $n$ non-final (there is at most one such transition),
then establish a path $\stackrel{a_1...a_{k-1}}{\leadsto}n$ in $N$, return the first state in this path, and stop.\vspace{1mm}
\item[4.3]
Add non-final states $n_1,...,n_k$ and transitions $n_1\!\stackrel{a_1}{\rightarrow}\!\cdots n_k\!\stackrel{a_k}{\rightarrow}\!z$ to $N$,
and return $n_1$.
\end{itemize}

The idea behind the construction of $N$ is as follows. Given a transition
$\theta=q\stackrel{\sigma/\tau}{\rightarrow}r$ in $P_0$, pushing $\tau$ onto the stack and moving to state $r$
corresponds to a path $y\stackrel{\tau^R}{\leadsto}r$ in $N$. A state $x$ in $N$ can jump
to the first state in this path if there is a path $x\stackrel{\sigma^R}{\Rightarrow}q$ in $N$,
because then we can reach $q$ from $x$ by pushing $\sigma$ onto the stack.
By executing $\theta$, we pop $\sigma$ from the stack, leading back to $x$, then jump to $y$, and push $\tau$ onto the stack
via the path $y\stackrel{\tau^R}{\leadsto}r$.
This jump is captured in $N$ by an $\eps$-transition from (every possible) $x$ to $y$.
To reduce the number of $\eps$-transitions in $N$, we only consider those $x$ with a path $x\stackrel{\sigma^R}{\Rightarrow}q$
in $N$ that does not start with an $\eps$-transition.

For each transition in $P_0$ at most one path is established in $N$, and for the rest $N$ consists of
$\eps$-transitions (between states in such a path), so the construction of $N$ always terminates.
The set $U_1$ returned at the end contains exactly the unreachable transitions in $P_0$.
A proof of this fact is presented at the end of this section.

We give an example construction of nfa $N$ from a pda.
As usual, the initial state in pda's and nfa's is drawn with an incoming arrow, and final states as a double circle.

\begin{example}
\label{exa:nfa1}
Consider the following pda $P_0$.

\vspace{3mm}

\centerline{\begin{picture}(0,0)%
\includegraphics{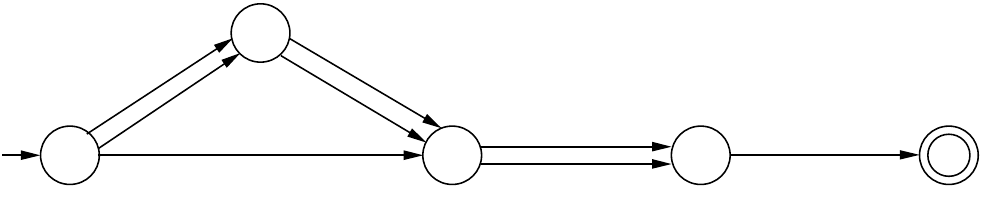}%
\end{picture}%
\setlength{\unitlength}{3947sp}%
\begingroup\makeatletter\ifx\SetFigFont\undefined%
\gdef\SetFigFont#1#2#3#4#5{%
  \reset@font\fontsize{#1}{#2pt}%
  \fontfamily{#3}\fontseries{#4}\fontshape{#5}%
  \selectfont}%
\fi\endgroup%
\begin{picture}(4704,1006)(2314,-760)
\put(2650,-536){\makebox(0,0)[b]{\smash{{\SetFigFont{10}{12.0}{\rmdefault}{\mddefault}{\updefault}{\color[rgb]{0,0,0}$q_0$}%
}}}}
\put(3565, 57){\makebox(0,0)[b]{\smash{{\SetFigFont{10}{12.0}{\rmdefault}{\mddefault}{\updefault}{\color[rgb]{0,0,0}$q_1$}%
}}}}
\put(3516,-633){\makebox(0,0)[b]{\smash{{\SetFigFont{10}{12.0}{\rmdefault}{\mddefault}{\updefault}{\color[rgb]{0,0,0}$\eps/da$}%
}}}}
\put(3957,-322){\makebox(0,0)[rb]{\smash{{\SetFigFont{10}{12.0}{\rmdefault}{\mddefault}{\updefault}{\color[rgb]{0,0,0}$\eps/d$}%
}}}}
\put(3080,-134){\makebox(0,0)[rb]{\smash{{\SetFigFont{10}{12.0}{\rmdefault}{\mddefault}{\updefault}{\color[rgb]{0,0,0}$\eps/a$}%
}}}}
\put(3142,-322){\makebox(0,0)[lb]{\smash{{\SetFigFont{10}{12.0}{\rmdefault}{\mddefault}{\updefault}{\color[rgb]{0,0,0}$\eps/b$}%
}}}}
\put(4485,-536){\makebox(0,0)[b]{\smash{{\SetFigFont{10}{12.0}{\rmdefault}{\mddefault}{\updefault}{\color[rgb]{0,0,0}$q_2$}%
}}}}
\put(5045,-696){\makebox(0,0)[b]{\smash{{\SetFigFont{10}{12.0}{\rmdefault}{\mddefault}{\updefault}{\color[rgb]{0,0,0}$db/\eps$}%
}}}}
\put(5045,-385){\makebox(0,0)[b]{\smash{{\SetFigFont{10}{12.0}{\rmdefault}{\mddefault}{\updefault}{\color[rgb]{0,0,0}$ca/\eps$}%
}}}}
\put(4112,-134){\makebox(0,0)[lb]{\smash{{\SetFigFont{10}{12.0}{\rmdefault}{\mddefault}{\updefault}{\color[rgb]{0,0,0}$\eps/c$}%
}}}}
\put(6216,-422){\makebox(0,0)[b]{\smash{{\SetFigFont{10}{12.0}{\rmdefault}{\mddefault}{\updefault}{\color[rgb]{0,0,0}$b_0/\eps$}%
}}}}
\put(5678,-536){\makebox(0,0)[b]{\smash{{\SetFigFont{10}{12.0}{\rmdefault}{\mddefault}{\updefault}{\color[rgb]{0,0,0}$q_3$}%
}}}}
\put(6869,-536){\makebox(0,0)[b]{\smash{{\SetFigFont{10}{12.0}{\rmdefault}{\mddefault}{\updefault}{\color[rgb]{0,0,0}$q_f$}%
}}}}
\end{picture}%
}

\vspace{3mm}

\noindent
We have taken the liberty to omit the state $q_e$ from $P_0$,
to keep the example small, and since the state $q_3$ is always reached with the stack $b_0$.

To determine the reachable transitions in $P_0$, the following nfa $N$ is constructed.

\vspace{4mm}

\centerline{\begin{picture}(0,0)%
\includegraphics{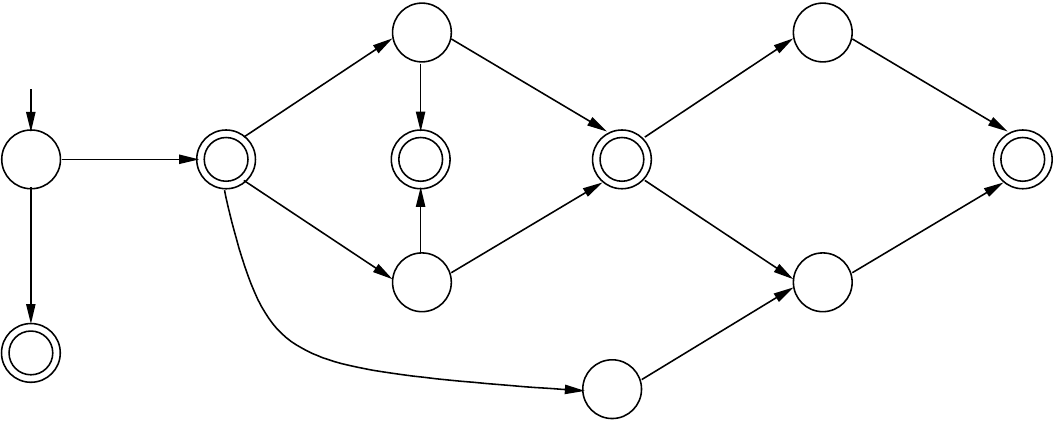}%
\end{picture}%
\setlength{\unitlength}{3947sp}%
\begingroup\makeatletter\ifx\SetFigFont\undefined%
\gdef\SetFigFont#1#2#3#4#5{%
  \reset@font\fontsize{#1}{#2pt}%
  \fontfamily{#3}\fontseries{#4}\fontshape{#5}%
  \selectfont}%
\fi\endgroup%
\begin{picture}(5059,2010)(967,-3338)
\put(2053,-2134){\makebox(0,0)[b]{\smash{{\SetFigFont{10}{12.0}{\rmdefault}{\mddefault}{\updefault}{\color[rgb]{0,0,0}$q_0$}%
}}}}
\put(2466,-3132){\makebox(0,0)[rb]{\smash{{\SetFigFont{10}{12.0}{\rmdefault}{\mddefault}{\updefault}{\color[rgb]{0,0,0}$\eps$}%
}}}}
\put(4347,-3068){\makebox(0,0)[lb]{\smash{{\SetFigFont{10}{12.0}{\rmdefault}{\mddefault}{\updefault}{\color[rgb]{0,0,0}$a$}%
}}}}
\put(3906,-3230){\makebox(0,0)[b]{\smash{{\SetFigFont{10}{12.0}{\rmdefault}{\mddefault}{\updefault}{\color[rgb]{0,0,0}$n_3$}%
}}}}
\put(4432,-1721){\makebox(0,0)[rb]{\smash{{\SetFigFont{10}{12.0}{\rmdefault}{\mddefault}{\updefault}{\color[rgb]{0,0,0}$\eps$}%
}}}}
\put(4917,-2723){\makebox(0,0)[b]{\smash{{\SetFigFont{10}{12.0}{\rmdefault}{\mddefault}{\updefault}{\color[rgb]{0,0,0}$n_4$}%
}}}}
\put(4432,-2555){\makebox(0,0)[rb]{\smash{{\SetFigFont{10}{12.0}{\rmdefault}{\mddefault}{\updefault}{\color[rgb]{0,0,0}$\eps$}%
}}}}
\put(5877,-2134){\makebox(0,0)[b]{\smash{{\SetFigFont{10}{12.0}{\rmdefault}{\mddefault}{\updefault}{\color[rgb]{0,0,0}$q_2$}%
}}}}
\put(5474,-1721){\makebox(0,0)[lb]{\smash{{\SetFigFont{10}{12.0}{\rmdefault}{\mddefault}{\updefault}{\color[rgb]{0,0,0}$c$}%
}}}}
\put(5474,-2555){\makebox(0,0)[lb]{\smash{{\SetFigFont{10}{12.0}{\rmdefault}{\mddefault}{\updefault}{\color[rgb]{0,0,0}$d$}%
}}}}
\put(4917,-1524){\makebox(0,0)[b]{\smash{{\SetFigFont{10}{12.0}{\rmdefault}{\mddefault}{\updefault}{\color[rgb]{0,0,0}$n_5$}%
}}}}
\put(2507,-1721){\makebox(0,0)[rb]{\smash{{\SetFigFont{10}{12.0}{\rmdefault}{\mddefault}{\updefault}{\color[rgb]{0,0,0}$\eps$}%
}}}}
\put(2993,-2723){\makebox(0,0)[b]{\smash{{\SetFigFont{10}{12.0}{\rmdefault}{\mddefault}{\updefault}{\color[rgb]{0,0,0}$n_2$}%
}}}}
\put(3051,-2481){\makebox(0,0)[lb]{\smash{{\SetFigFont{10}{12.0}{\rmdefault}{\mddefault}{\updefault}{\color[rgb]{0,0,0}$\eps$}%
}}}}
\put(2507,-2555){\makebox(0,0)[rb]{\smash{{\SetFigFont{10}{12.0}{\rmdefault}{\mddefault}{\updefault}{\color[rgb]{0,0,0}$\eps$}%
}}}}
\put(3953,-2134){\makebox(0,0)[b]{\smash{{\SetFigFont{10}{12.0}{\rmdefault}{\mddefault}{\updefault}{\color[rgb]{0,0,0}$q_1$}%
}}}}
\put(3550,-1721){\makebox(0,0)[lb]{\smash{{\SetFigFont{10}{12.0}{\rmdefault}{\mddefault}{\updefault}{\color[rgb]{0,0,0}$a$}%
}}}}
\put(3550,-2555){\makebox(0,0)[lb]{\smash{{\SetFigFont{10}{12.0}{\rmdefault}{\mddefault}{\updefault}{\color[rgb]{0,0,0}$b$}%
}}}}
\put(2993,-1524){\makebox(0,0)[b]{\smash{{\SetFigFont{10}{12.0}{\rmdefault}{\mddefault}{\updefault}{\color[rgb]{0,0,0}$n_1$}%
}}}}
\put(3051,-1785){\makebox(0,0)[lb]{\smash{{\SetFigFont{10}{12.0}{\rmdefault}{\mddefault}{\updefault}{\color[rgb]{0,0,0}$\eps$}%
}}}}
\put(2987,-2134){\makebox(0,0)[b]{\smash{{\SetFigFont{10}{12.0}{\rmdefault}{\mddefault}{\updefault}{\color[rgb]{0,0,0}$q_3$}%
}}}}
\put(1565,-2042){\makebox(0,0)[b]{\smash{{\SetFigFont{10}{12.0}{\rmdefault}{\mddefault}{\updefault}{\color[rgb]{0,0,0}$b_0$}%
}}}}
\put(1117,-2133){\makebox(0,0)[b]{\smash{{\SetFigFont{10}{12.0}{\rmdefault}{\mddefault}{\updefault}{\color[rgb]{0,0,0}$m_0$}%
}}}}
\put(1156,-2555){\makebox(0,0)[lb]{\smash{{\SetFigFont{10}{12.0}{\rmdefault}{\mddefault}{\updefault}{\color[rgb]{0,0,0}$\eps$}%
}}}}
\put(1116,-3063){\makebox(0,0)[b]{\smash{{\SetFigFont{10}{12.0}{\rmdefault}{\mddefault}{\updefault}{\color[rgb]{0,0,0}$q_f$}%
}}}}
\end{picture}%
}

\begin{itemize}
\item
Initially $N$ consists of $m_0\stackrel{b_0}{\rightarrow}q_0$.\vspace{1mm}
\item
First the paths $q_0\stackrel{\eps}{\rightarrow}n_1\stackrel{a}{\rightarrow}q_1$ and $q_0\stackrel{\eps}{\rightarrow}n_2\stackrel{b}{\rightarrow}q_1$
and $q_0\stackrel{\eps}{\rightarrow}n_3\stackrel{a}{\rightarrow}n_4\stackrel{d}{\rightarrow}q_2$ are added to $N$, by the transitions
$q_0\stackrel{\eps/a}{\rightarrow}q_1$ and $q_0\stackrel{\eps/b}{\rightarrow}q_1$ and $q_0\stackrel{\eps/da}{\rightarrow}q_2$, respectively, in $P_0$.
These transitions are deleted from $U_1$.\vspace{1mm}
\item
Next the paths $q_1\stackrel{\eps}{\rightarrow}n_5\stackrel{c}{\rightarrow}q_2$ and $q_1\stackrel{\eps}{\rightarrow}n_4$ are added
to $N$, by the transitions $q_1\stackrel{\eps/c}{\rightarrow}q_2$ and $q_1\stackrel{\eps/d}{\rightarrow}q_2$, respectively, in $P_0$,
which are deleted from $U_1$.\vspace{1mm}
\item
Next the transitions $n_1\stackrel{\eps}{\rightarrow}q_3$ and $n_2\stackrel{\eps}{\rightarrow}q_3$ are added to $N$,
by the transitions $q_2\!\stackrel{ca/\eps}{\rightarrow}\!q_3$ and $q_2\!\stackrel{db/\eps}{\rightarrow}\!q_3$, respectively, in $P_0$,
which are deleted from $U_1$.\vspace{1mm}
\item
Finally the transition $m_0\stackrel{\eps}{\rightarrow}q_f$ is added to $N$,
by the transition $q_3\stackrel{b_0/\eps}{\rightarrow}q_f$ in $P_0$,
which is deleted from $U_1$.
\end{itemize}
Since all transitions in $P_0$ are applied in the construction of $N$, they are all reachable. That is, at the end $U_1=\emptyset$,
and $P_1$ coincides with $P_0$.
\end{example}

\paragraph{Correctness proof}
The following two properties of $N$, which follow immediately from its construction, give insight into the structure of $N$.
In particular, Lemma \ref{lem:nfa1} implies that in $N$, the outgoing transitions of a final state always carry the label $\eps$,
while each non-final state has exactly one outgoing transition with a label from $\Gamma$.
%% The lemmas are employed in the correctness proof of the {\em forward} procedure.

\begin{lemma}
\label{lem:nfa-reachable}
For each state $x$ in $N$, there is a path $m_0\stackrel{\sigma}{\Rightarrow}x$ in $N$, for some $\sigma$.
\end{lemma}

\begin{lemma}
\label{lem:nfa1}
For each state $x$ in $N$ there is exactly one path $x\stackrel{\sigma}{\leadsto}q$ in $N$, for some $\sigma$, $q$.
\end{lemma}

\begin{proof}
We prove both lemmas in one go, by induction on the construction of $N$.

Initially they hold trivially, because then $N$ only consists of the transition $m_0\stackrel{b_0}{\rightarrow}q_0$.

When a path $y\stackrel{\tau}{\leadsto}q$ is added to $N$, then also a transition $x\stackrel{\eps}{\rightarrow}y$
is added, where $x$ was already present in $N$. Since by induction $m_0\stackrel{\rho}{\Rightarrow}x$ for some $\rho$,
clearly Lemma \ref{lem:nfa-reachable} still holds in the extended nfa $N$.
Moreover, since only (part of) the path $y\stackrel{\tau}{\leadsto}q$ is added to $N$, together with some
$\eps$-transitions to the first state in this path, clearly Lemma \ref{lem:nfa1} still holds in the extended nfa $N$.
\qed
\end{proof}

\noindent
The following lemma and proposition are corner stones in the correctness proof.

\begin{lemma}
\label{lem:strengthening}
If there are paths $x\stackrel{\sigma^R}{\leadsto}q$ and $x\stackrel{\tau^R}{\Rightarrow}r$ in $N$, then there is a run $(q,\sigma)\rightarrow^*(r,\tau)$ of $P_1$.
\end{lemma}

\begin{proof}
The lemma is proved by induction on the path $x\stackrel{\tau^R}{\Rightarrow}r$. A well-founded partial ordering on paths in $N$ is defined as follows.
Suppose that during the construction of $N$, each $\eps$-transition is at its creation provided with a sequence number, being one
higher than the previously created $\eps$-transition (the first created $\eps$-transition gets sequence number $0$).
Now $y_1\stackrel{\rho_1}{\Rightarrow}s_1$ is defined to be smaller than $y_2\stackrel{\rho_2}{\Rightarrow}s_2$ if:

\vspace*{-2mm}
\begin{itemize}
\item[(i)]
either $y_2\stackrel{\rho_2}{\Rightarrow}s_2$ contains an $\eps$-transition with a higher sequence number than any of the
$\eps$-transitions in $y_1\stackrel{\rho_1}{\Rightarrow}s_1$;
\item[(ii)]
or $y_2\stackrel{\rho_2}{\Rightarrow}s_2$ contains $y_1\stackrel{\rho_1}{\Rightarrow}s_1$ as a strict subsequence.
\end{itemize}

\vspace*{-2mm}

\noindent
In the base case of the induction, the path $x\stackrel{\tau^R}{\Rightarrow}r$ consists of zero transitions, so that $\tau=\eps$ and $x=r$.
By Lemma \ref{lem:nfa1}, the path $r\stackrel{\sigma^R}{\leadsto}q$ in $N$ implies $\sigma=\eps$ and $r=q$. And trivially there is a run
$(q,\eps)\rightarrow^*(q,\eps)$ of $P_1$.

In the inductive case, the path $x\stackrel{\tau^R}{\Rightarrow}r$ consists of one or more transitions.
We distinguish two cases, depending on whether the path $x\stackrel{\tau^R}{\Rightarrow}r$ starts with an $\eps$-transition.

\vspace{1mm}

\noindent
{\sc Case 1}: The path $x\stackrel{\tau^R}{\Rightarrow}r$ is of the form $x\stackrel{a}{\rightarrow}y\stackrel{\tau'^R}{\Rightarrow}r$, with $\tau=\tau' a$.
By Lemma \ref{lem:nfa1}, $x$ is non-final, and $x\stackrel{\sigma^R}{\leadsto}q$ is of the form $x\stackrel{a}{\rightarrow}y\stackrel{\sigma'^R}{\leadsto}q$,
with $\sigma=\sigma' a$.
By (ii), $y\stackrel{\tau'^R}{\Rightarrow}r$ is smaller than $x\stackrel{\tau^R}{\Rightarrow}r$.
Since there are paths $y\stackrel{\sigma'^R}{\leadsto}q$ and $y\stackrel{\tau'^R}{\Rightarrow}r$ in $N$,
by induction, there is a run $(q,\sigma')\rightarrow^*(r,\tau')$ of $P_1$.
So there is a run $(q,\sigma'a)\rightarrow^*(r,\tau'a)$ of $P_1$.

\vspace{1mm}

\noindent
{\sc Case 2}: The path $x\stackrel{\tau^R}{\Rightarrow}r$ is of the form $x\stackrel{\eps}{\rightarrow}y\stackrel{\tau^R}{\Rightarrow}r$.
Suppose the transition $x\stackrel{\eps}{\rightarrow}y$ was created in $N$ due to a transition $s\stackrel{\upsilon/\zeta}{\rightarrow}t$ in $P_1$.
Then there must be paths $x\stackrel{\upsilon^R}{\Rightarrow}s$ and $y\stackrel{\zeta^R}{\leadsto}t$ in $N$.
By (i), $x\stackrel{\upsilon^R}{\Rightarrow}s$ is smaller than $x\stackrel{\tau^R}{\Rightarrow}r$, because
$x\stackrel{\eps}{\rightarrow}y$ has a higher sequence number than any of the $\eps$-transitions in the path $x\stackrel{\upsilon^R}{\Rightarrow}s$.
Since $x\stackrel{\sigma^R}{\leadsto}q$ and $x\stackrel{\upsilon^R}{\Rightarrow}s$,
by induction, there is a run $(q,\sigma)\rightarrow^*(s,\upsilon)$ of $P_1$.
The transition $s\stackrel{\upsilon/\zeta}{\rightarrow}t$ in $P_1$ gives rise to the move $(s,\upsilon)\rightarrow(t,\zeta)$.
By (ii), $y\stackrel{\tau^R}{\Rightarrow}r$ is smaller than $x\stackrel{\tau^R}{\Rightarrow}r$.
Since $y\stackrel{\zeta^R}{\leadsto}t$ and $y\stackrel{\tau^R}{\Rightarrow}r$,
by induction, there is a run $(t,\zeta)\rightarrow^*(r,\tau)$ of $P_1$.
Concluding, there is a run $(q,\sigma)\rightarrow^*(r,\tau)$ of $P_1$.
\qed
\end{proof}

\begin{proposition}
\label{prop:nfa1}
There is a path $m_0\stackrel{\sigma^R}{\Rightarrow}r$ in $N$ if, and only if, $(r,\sigma)$ is reachable in $P_0$.
\end{proposition}

\begin{proof}
($\Rightarrow$)
The transition $m_0\stackrel{b_0}{\rightarrow}q_0$ is in $N$, and by assumption there is a path $m_0\stackrel{\sigma^R}{\Rightarrow}r$ in $N$.
So by Lemma \ref{lem:strengthening}, there is a run $(q_0,b_0)\rightarrow^*(r,\sigma)$ of $P_1$, and so of $P_0$.

\vspace{2mm}

\noindent
($\Leftarrow$) By induction on the number of moves in a run $(q_0,b_0)\rightarrow^*(r,\sigma)$ of $P_0$.
In the base case, no transition is applied, so $r=q_0$ and $\sigma=b_0$.
The transition $m_0\stackrel{b_0}{\rightarrow}q_0$ is in $N$.

In the inductive case, suppose $q\stackrel{\tau/\upsilon}{\rightarrow}r$ is the last transition applied in the run $(q_0,b_0)\rightarrow^*(r,\sigma)$.
Then $\sigma=\upsilon\rho$ for some $\rho$, and there is a run $(q_0,b_0)\rightarrow^*(q,\tau\rho)$ of $P_0$. Since this run
takes one move less than the one to $(r,\sigma)$, by induction there is a path $m_0\!\!\stackrel{(\tau\rho)^R}{\Rightarrow}\!\!q$ in $N$.
This path splits into $m_0\stackrel{\rho^R}{\Rightarrow}x\stackrel{\tau^R}{\Rightarrow}q$ in $N$, where we choose $x$ such that
$x\stackrel{\tau^R}{\Rightarrow}q$ does not start with an $\eps$-transition.
In view of the transition $q\stackrel{\tau/\upsilon}{\rightarrow}r$ in $P_0$ and the path $x\stackrel{\tau^R}{\Rightarrow}q$ in $N$,
during the construction of $N$, a path $y\stackrel{\upsilon^R}{\leadsto}r$ was created, together with
the transition $x\stackrel{\eps}{\rightarrow}y$. Concluding, there is a path $m_0\stackrel{\rho^R}{\Rightarrow}x\stackrel{\eps}{\rightarrow}y\stackrel{\upsilon^R}{\leadsto}r$,
so $m_0\stackrel{\sigma^R}{\Rightarrow}r$, in $N$.
\qed
\end{proof}

\begin{theorem}
\label{thm:nfa1}
The returned set $U_1$ consists of the unreachable transitions in $P_0$.
\end{theorem}

\begin{proof}
Suppose $\theta=q\stackrel{\sigma/\tau}{\rightarrow}r$ in $P_0$ is not in $U_1$. Then during the construction of $N$,
$\theta$ was used in the creation of a path $y\stackrel{\tau^R}{\leadsto}r$, together with one or more transitions
$x\stackrel{\eps}{\rightarrow}y$. We choose one such $x$. The construction requires that there is a path $x\stackrel{\sigma^R}{\Rightarrow}q$
in $N$. And by Lemma \ref{lem:nfa-reachable}, $m_0\stackrel{\upsilon^R}{\Rightarrow}x$ for some $\upsilon$. So by Proposition \ref{prop:nfa1},
$(q,\sigma\upsilon)$ is reachable in $P_0$. In this configuration, $\theta$ can be applied, to reach $(r,\tau\upsilon)$.
So $\theta$ is reachable in $P_0$.

Vice versa, suppose $\theta$ is reachable in $P_0$. Then a configuration $(q,\sigma\rho)$ is reachable in $P_0$, for some $\rho$.
So by Proposition \ref{prop:nfa1} there is a path $m_0\!\!\stackrel{(\sigma\rho)^R}{\Rightarrow}\!\!q$ in $N$.
This path splits into $m_0\stackrel{\rho^R}{\Rightarrow}x\stackrel{\sigma^R}{\Rightarrow}q$ in $N$, where we choose $x$ such that
$x\stackrel{\sigma^R}{\Rightarrow}q$ does not start with an $\eps$-transition.
In view of $\theta$, a path $y\stackrel{\tau^R}{\leadsto}r$ and a transition $x\stackrel{\eps}{\rightarrow}y$ were added to $N$.
And as a result, $\theta$ was deleted from $U_1$.
\qed
\end{proof}

\paragraph{Complexity analysis}
Let $Q$ be the number of states and $T$ the number of transitions in the pda $P_0$.
In the analysis of the worst-case time complexity of the algorithm we assume that the number of elements popped from and
pushed onto the stack in one transition, as well as the size of the stack alphabet, are bounded by some constant.
Then the nfa $N$ contains at most $O(Q)$ states.

Building $N$ takes at most $O(Q^4T)$: During a run of the {\em forward} procedure over the transitions in $P_0$,
at most $T$ times (once for each transition in $P_0$), in step 2 backward scans over the $\eps$-transitions are performed,
which take at most $O(Q^2)$ (because there are at most $O(Q^2)$ $\eps$-transitions);
and there are at most $O(Q^2)$ runs of the {\em forward} procedure over the transitions in $P_0$ (because $N$ contains at most $O(Q^2)$ transitions).

\subsection{Detecting which transitions can lead to the final state}
\label{sec:backward}

If the transition $m_0\stackrel{\eps}{\rightarrow}q_f$ is not in $N$, then by Proposition \ref{prop:nfa1} the language accepted by $P_1$
is empty. So then all transitions in $P_1$ are reported as useless.

If the transition $m_0\stackrel{\eps}{\rightarrow}q_f$ is in $N$, then
the set $U_2$ of useless transitions in $P_1$ is constructed by running the following {\em backward} procedure
over $\eps$-transitions in $N$ that are in a set $E\backslash F$; at the start of such a run an $\eps$-transition from $E\backslash F$ is
copied to the set $F$, while on the other hand during the run $\eps$-transitions from $N$ may be added to $E$.

Initially $U_2$, as an overapproximation, contains all transitions in $P_1$, $E=\{m_0\stackrel{\eps}{\rightarrow}q_f\}$
and $F=\emptyset$. We recall from step 2 of the {\em forward} procedure that the set $S_{q,\sigma}$ equals $\{q\}$ if $\sigma=\eps$,
or the states $n$ for which there exists a path $n\stackrel{a}{\rightarrow}y\stackrel{\sigma'^R}{\Longrightarrow}q$ in $N$ if $\sigma=\sigma'a$.
These sets have already been computed in (the last run of) the {\em forward} procedure.

\vspace{4mm}

\noindent
Procedure {\em backward}: While $E\backslash F\neq\emptyset$ do:
\begin{enumerate}
\item[1.~~]
Pick an $x\stackrel{\eps}{\rightarrow}y\in E\backslash F$ and add it to $F$.\vspace{1mm}
\item[2.~~]
Find the path $y\stackrel{\tau^R}{\leadsto}r$ in $N$ (for some $\tau$, $r$); since $r$ denotes a final state, according to Lemma \ref{lem:nfa1}, exactly one such path exists.\vspace{1mm}
\item[3.~~]
For each transition $\theta=q\stackrel{\sigma/\tau}{\rightarrow}r$ in $P_1$ (for any $q$, $\sigma$) do:\vspace{1mm}
\item[3.1]
If $x\not\in S_{q,\sigma}$, then stop this iteration step (i.e., return to step 3).\vspace{1mm}
\item[3.2]
If $\theta\in U_2$, then delete $\theta$ from $U_2$.\vspace{1mm}
\item[3.3]
If $\sigma=\sigma' a$ (i.e., $\sigma\neq\eps$), then add to $E$ the $\eps$-transitions that occur in any path
$x\stackrel{a}{\rightarrow}z\stackrel{\sigma'^R}{\Rightarrow}q$ in $N$.
\end{enumerate}

\vspace{1mm}

\noindent
At the end, return the set $U_2$ of useless transitions in $P_1$. The transitions in $U_1\cup U_2$ that stem from the original pda $P$
(i.e., not those from the preprocessing step in which $q_e$ and $q_f$ were added) are the useless transitions in $P$, so can be culled without
changing the associated language.

The idea behind the {\em backward} procedure is that a transition $x\stackrel{\eps}{\rightarrow}y$ in $N$ is added to $E$ when
we are certain that, for some $\rho$, $\upsilon$ and $s$, there is a path
$m_0\stackrel{\rho^R}{\Rightarrow}x\stackrel{\eps}{\rightarrow}y\stackrel{\upsilon^R}{\Rightarrow}s$ in $N$
and a run $(s,\upsilon\rho)\rightarrow^*(q_f,\eps)$ of $P_1$.
Each transition $x\stackrel{\eps}{\rightarrow}y$ in $E$ may in turn give rise to adding $\eps$-transitions in $N$ to $E$,
and removing transitions from $U_2$. Namely, by Lemma \ref{lem:nfa1} there is exactly one path $y\stackrel{\tau^R}{\leadsto}r$ in $N$ (for some $\tau$, $r$).
For each transition $\theta=q\stackrel{\sigma/\tau}{\rightarrow}r$ (for any $q$, $\sigma$) in $P_1$,
all $\eps$-transitions in any path $x\stackrel{\sigma^R}{\Rightarrow}q$ in $N$
can be added to $E$. The reason is that there is a path $m_0\stackrel{\rho^R}{\Rightarrow}x\stackrel{\sigma^R}{\Rightarrow}q$ in $N$,
as well as a run $(q,\sigma\rho)\rightarrow^*(q_f,\eps)$ of $P_1$: The transition $\theta$ gives rise to the move $(q,\sigma\rho)\rightarrow(r,\tau\rho)$;
in view of the paths $y\stackrel{\tau^R}{\leadsto}r$ and $y\stackrel{\upsilon^R}{\Rightarrow}s$ in $N$,
by Lemma \ref{lem:strengthening} there is a run $(r,\tau\rho)\rightarrow^*(s,\upsilon\rho)$ of $P_1$; and
by assumption there is a run $(s,\upsilon\rho)\rightarrow^*(q_f,\eps)$ of $P_1$. So if there is a path $x\stackrel{\sigma^R}{\Rightarrow}q$ in $N$,
then $\theta$ is useful: In view of the path $m_0\stackrel{\rho^R}{\Rightarrow}x\stackrel{\sigma^R}{\Rightarrow}q$ in $N$, by Proposition \ref{prop:nfa1},
the configuration $(q,\sigma\rho)$ is reachable in $P_0$; and we argued there is a run $(q,\sigma\rho)\rightarrow^*(q_f,\eps)$ of $P_1$,
which starts with an application of $\theta$. Hence $\theta$ can be removed from $U_2$.
To try and avoid adding the same $\eps$-transition to $E$ more than once,
we only consider those paths $x\stackrel{\sigma^R}{\Rightarrow}q$ in $N$ that do not start with an $\eps$-transition.

The {\em backward} procedure is executed only a finite number of times, as it is performed at most once for each $\eps$-transition in $N$.
The set $U_2$ returned at the end contains exactly the useless transitions in $P_1$.
The corresponding theorem is presented at the end of this section, together with two propositions needed in its proof.

\begin{example}
\label{exa:nfa2}
We perform the {\em backward} procedure on the nfa $N$ from Example \ref{exa:nfa1}.
\begin{itemize}
\item
Initially $E=\{m_0\stackrel{\eps}{\rightarrow}q_f\}$ and $F=\emptyset$.\vspace{2mm}
\item
$m_0\stackrel{\eps}{\rightarrow}q_f$ is added to $F$; then $\tau=\eps$ and $r=q_f$.
Since $S_{q_3,b_0}=\{m_0\}$, the transition $q_3\stackrel{b_0/\eps}{\rightarrow}q_f$ is deleted from $U_2$,
and the $\eps$-transitions in paths $m_0\stackrel{b_0}{\rightarrow}z\stackrel{\eps}{\Rightarrow}q_3$ in $N$ are added to $E$:
$q_0\stackrel{\eps}{\rightarrow}n_1\stackrel{\eps}{\rightarrow}q_3$ and $q_0\stackrel{\eps}{\rightarrow}n_2\stackrel{\eps}{\rightarrow}q_3$.\vspace{2mm}
\item
$q_0\stackrel{\eps}{\rightarrow}n_1$ is added to $F$; then $\tau=a$ and $r=q_1$. Since $S_{q_0,\eps}=\{q_0\}$, the transition $q_0\stackrel{\eps/a}{\rightarrow}q_1$ is deleted from $U_2$.\vspace{2mm}
\item
$q_0\stackrel{\eps}{\rightarrow}n_2$ is added to $F$; then $\tau=b$ and $r=q_1$. Since $S_{q_0,\eps}=\{q_0\}$, the transition $q_0\stackrel{\eps/b}{\rightarrow}q_1$ is deleted from $U_2$.\vspace{2mm}
\item
$n_1\stackrel{\eps}{\rightarrow}q_3$ is added to $F$; then $\tau=\eps$ and $r=q_3$.
Since $S_{q_2,ca}=\{n_1\}$, the transition $q_2\stackrel{ca/\eps}{\rightarrow}q_3$ is deleted from $U_2$,
and the $\eps$-transitions in paths $n_1\stackrel{a}{\rightarrow}z\stackrel{c}{\Rightarrow}q_2$ in $N$ are added to $E$: $q_1\stackrel{\eps}{\rightarrow}n_5$.\vspace{2mm}
\item
$n_2\stackrel{\eps}{\rightarrow}q_3$ is added to $F$; then $\tau=\eps$ and $r=q_3$.
Since $S_{q_2,db}=\{n_2\}$, the transition $q_2\stackrel{db/\eps}{\rightarrow}q_3$ is deleted from $U_2$,
and the $\eps$-transitions in paths $n_2\stackrel{b}{\rightarrow}z\stackrel{d}{\Rightarrow}q_2$ in $N$ are added to $E$: $q_1\stackrel{\eps}{\rightarrow}n_4$.\vspace{2mm}
\item
$q_1\stackrel{\eps}{\rightarrow}n_5$ is added to $F$; then $\tau=c$ and $r=q_2$. Since $S_{q_1,\eps}=\{q_1\}$, the transition $q_1\stackrel{\eps/c}{\rightarrow}q_2$ is deleted from $U_2$.\vspace{2mm}
\item
$q_1\stackrel{\eps}{\rightarrow}n_4$ is added to $F$; then $\tau=d$ and $r=q_2$. Since $S_{q_1,\eps}=\{q_1\}$, the transition $q_1\stackrel{\eps/d}{\rightarrow}q_2$ is deleted from $U_2$.
\end{itemize}
At the end, $U_2$ consists of $q_0\stackrel{\eps/da}{\rightarrow}q_2$. So this transition is useless in $P_1$.
\end{example}

Example \ref{exa:nfa2} shows that in step 2 of the {\em backward} procedure,
the path $y\stackrel{\tau^R}{\leadsto}r$ cannot be replaced by all paths $y\stackrel{\tau^R}{\Rightarrow}r$ in $N$.
Else $q_0\stackrel{\eps/da}{\rightarrow}q_2$ would be erroneously deleted from $U_2$, in view of
the transition $q_0\stackrel{\eps}{\rightarrow}n_1$ in $E$ and
the path $n_1\stackrel{a}{\rightarrow}q_1\stackrel{\eps}{\rightarrow}n_4\stackrel{d}{\rightarrow}q_2$ in $N$.

\paragraph{Correctness proof}
The following proposition is needed to show that only useful transitions in $P_1$ are deleted from $U_2$ during runs
of the {\em backward} procedure.

\vspace*{-2mm}

\begin{proposition}
\label{prop:nfa2}
Let $x\stackrel{\eps}{\rightarrow}y$ be a transition in $E$, and $y\stackrel{\tau^R}{\leadsto}r$ a path in $N$.
Then there is a path $m_0\stackrel{\rho^R}{\Rightarrow}x$ in $N$, for some $\rho$, such that
there is a run $(r,\tau\rho)\rightarrow^*(q_f,\eps)$ of $P_1$.
\end{proposition}

\begin{proof}
By induction on the construction of $E$.
In the base case, $E=\{m_0\stackrel{\eps}{\rightarrow}q_f\}$. So $x=m_0$, $y=q_f$, $\tau=\eps$ and $r=q_f$.
So we can take $\rho=\eps$.

In the inductive case, suppose that due to a transition $z\stackrel{\eps}{\rightarrow}w$ in $E$, a path $w\stackrel{\upsilon^R}{\leadsto}s$ in $N$,
and a transition $\theta=q\stackrel{\sigma/\upsilon}{\rightarrow}s$ in $P_1$,
$\eps$-transitions in paths $z\stackrel{\sigma^R}{\Rightarrow}q$ in $N$ are added to $E$.
We show that the proposition holds for any transition $x\stackrel{\eps}{\rightarrow}y$ in any path $z\stackrel{\sigma^R}{\Rightarrow}q$;
say $z\stackrel{\sigma_2^R}{\Rightarrow}x\stackrel{\eps}{\rightarrow}y\stackrel{\sigma_1^R}{\Rightarrow}q$ with $\sigma=\sigma_1\sigma_2$.
Since $y\stackrel{\tau}{\leadsto}r$ and $y\stackrel{\sigma_1^R}{\Rightarrow}q$ are paths in $N$, by Lemma \ref{lem:strengthening},
there is a run $(r,\tau)\rightarrow^*(q,\sigma_1)$ of $P_1$. And by $\theta$, $(q,\sigma)\rightarrow(s,\upsilon)$ in $P_1$.
The transition $z\stackrel{\eps}{\rightarrow}w$ was already present in $E$ before the current extension of $E$.
Since $w\stackrel{\upsilon^R}{\leadsto}s$ is a path in $N$, by induction there is a path $m_0\stackrel{\rho^R}{\Rightarrow}z$ in $N$, for some $\rho$,
such that there is a run $(s,\upsilon\rho)\rightarrow^*(q_f,\eps)$ of $P_1$.
Concluding, there is a path $m_0\stackrel{\rho^R}{\Rightarrow}z\stackrel{\sigma_2^R}{\Rightarrow}x$ in $N$,
and a run $(r,\tau\sigma_2\rho)\rightarrow^*(q,\sigma_1\sigma_2\rho)\rightarrow(s,\upsilon\rho)\rightarrow^*(q_f,\eps)$ of $P_1$.
\qed
\end{proof}

The idea behind the next lemma is that if $(r,\tau\rho)$ is a reachable configuration of $P_1$,
and $\rho\neq\eps$, then any run $(r,\tau\rho)\rightarrow^*(q_f,\eps)$ of $P_1$ contains a move
$(q,\sigma\rho)\rightarrow(s,\upsilon\rho_2)$ in which a non-empty part $\rho_1$ of $\rho=\rho_1\rho_2$ is removed from the stack.

\begin{lemma}
\label{lem:nfa3}
Let $m_0\stackrel{\rho^R}{\Rightarrow}x\stackrel{\tau^R}{\Rightarrow}r$ be a path in $N$, and
$(r,\tau\rho)\rightarrow^*(q_f,\eps)$ a run of $P_1$, with $\rho\neq\eps$.
Then there is a path $x\stackrel{\sigma^R}{\Rightarrow}q$ in $N$ for some $\sigma$ and $q$, and
the path $m_0\stackrel{\rho^R}{\Rightarrow}x$ in $N$ splits into
$m_0\stackrel{\rho_2^R}{\Rightarrow}y\stackrel{\rho_1^R}{\Rightarrow}x$ for some $\rho_2$, $y$ and $\rho_1\neq\eps$,
where $y\stackrel{\rho_1^R}{\Rightarrow}x$ does not start with an $\eps$-transition, and
there is a path $y\stackrel{\eps}{\rightarrow}z\stackrel{\upsilon^R}{\leadsto}s$ in $N$ for some $z$, $\upsilon$ and $s$,
where $q\stackrel{\sigma\rho_1/\upsilon}{\rightarrow}s$ is a transition in $P_1$
that is applied to the configuration $(q,\sigma\rho)$ in the run $(r,\tau\rho)\rightarrow^*(q_f,\eps)$.
\end{lemma}

\begin{proof}
By induction on the number of moves in the run $(r,\tau\rho)\rightarrow^*(q_f,\eps)$.
Note that $\rho\neq\eps$ implies $r\neq q_f$. Let $\theta=r\stackrel{\zeta/\pi}{\rightarrow}t$ in $P_1$
be the first transition that is applied in the run $(r,\tau\rho)\rightarrow^*(q_f,\eps)$. We distinguish two cases:

\vspace{1mm}

\noindent
{\sc Case 1}: $\zeta=\tau\rho_1$ with $\rho=\rho_1\rho_2$, for some $\rho_1\neq\eps$ and $\rho_2$.
We take $\sigma=\tau$ and $q=r$. The path $m_0\stackrel{\rho^R}{\Rightarrow}x$ in $N$ splits into
$m_0\stackrel{\rho_2^R}{\Rightarrow}y\stackrel{\rho_1^R}{\Rightarrow}x$, where we choose $y$ such that
$y\stackrel{\rho_1^R}{\Rightarrow}x$ does not start with an $\eps$-transition.
Since there is a path $y\stackrel{\rho_1^R}{\Rightarrow}x\stackrel{\tau^R}{\Rightarrow}r$ in $N$, $\zeta=\tau\rho_1$,
and $\theta$ is in $P_1$, by the {\em forward} procedure, there is a path $y\stackrel{\eps}{\rightarrow}z\stackrel{\pi^R}{\leadsto}t$ in $N$.
So if we take $s=t$ and $\upsilon=\pi$, we are done.

\vspace{1mm}

\noindent
{\sc Case 2}: $\tau=\zeta\tau'$ for some $\tau'$. The path $x\stackrel{\tau^R}{\Rightarrow}r$ in $N$
splits into $x\stackrel{\tau'^R}{\Rightarrow}w\stackrel{\zeta^R}{\Rightarrow}r$, where we choose $w$ such that
$w\stackrel{\zeta^R}{\Rightarrow}r$ does not start with an $\eps$-transition.
Since $\theta$ is in $P_1$, by the {\em forward} procedure, there is a path $w\stackrel{\eps}{\rightarrow}v\stackrel{\pi^R}{\leadsto}t$ in $N$.
So there is a path $m_0\stackrel{\rho^R}{\Rightarrow}x\stackrel{(\pi\tau')^R}{\Rightarrow}t$ in $N$.
Moreover, the run $(r,\tau\rho)\rightarrow^*(q_f,\eps)$ is of the form
$(r,\tau\rho)\rightarrow(t,\pi\tau'\rho)\rightarrow^*(q_f,\eps)$.
By induction, there is a path $x\stackrel{\sigma^R}{\Rightarrow}q$ in $N$,
and the path $m_0\stackrel{\rho^R}{\Rightarrow}x$ in $N$ splits into $m_0\stackrel{\rho_2^R}{\Rightarrow}y\stackrel{\rho_1^R}{\Rightarrow}x$
with $\rho_1\neq\eps$,
where $y\stackrel{\rho_1^R}{\Rightarrow}x$ does not start with an $\eps$-transition, and
there is a path $y\stackrel{\eps}{\rightarrow}z\stackrel{\upsilon^R}{\leadsto}s$ in $N$,
where $q\stackrel{\sigma\rho_1/\upsilon}{\rightarrow}s$ is a transition in $P_1$ that is applied to the configuration $(q,\sigma\rho)$
in the run $(t,\pi\tau'\rho)\rightarrow^*(q_f,\eps)$.
\qed
\end{proof}

The following proposition is needed to show that all useful transitions in $P_1$ are eventually deleted from $U_2$.

\begin{proposition}
\label{prop:nfa3}
Let $m_0\stackrel{\rho^R}{\Rightarrow}x\stackrel{\eps}{\rightarrow}y\stackrel{\tau^R}{\Rightarrow}r$ be a path in $N$,
and $(r,\tau\rho)\rightarrow^*(q_f,\eps)$ a run of $P_1$. Then $x\stackrel{\eps}{\rightarrow}y$ is in $E$ when the
{\em backward} procedure terminates.
\end{proposition}

\begin{proof}
We apply induction on the length of $\rho$.
In the base case, $\rho=\eps$. Since $m_0\stackrel{b_0}{\rightarrow}q_0$ is a transition in $N$, and $q_e\stackrel{b_0/\eps}{\rightarrow}q_f$
is the only transition in $P_0$ with an occurrence of $b_0$, by the {\em forward} procedure, the only possible outgoing $\eps$-transition of
$m_0$ in $N$ is $m_0\stackrel{\eps}{\rightarrow}q_f$. Hence the path $m_0\stackrel{\eps}{\Rightarrow}x\stackrel{\eps}{\rightarrow}y\stackrel{\tau^R}{\Rightarrow}r$
implies that $x=m_0$, $y=q_f$, $\tau=\eps$ and $r=q_f$.
Since $m_0\stackrel{\eps}{\rightarrow}q_f$ is in $N$, initially $E=\{m_0\stackrel{\eps}{\rightarrow}q_f\}$.

In the inductive case, $\rho\neq\eps$. Since there is a path $m_0\stackrel{\rho^R}{\Rightarrow}x\stackrel{\eps}{\rightarrow}y\stackrel{\tau^R}{\Rightarrow}r$ in $N$
and a run $(r,\tau\rho)\rightarrow^*(q_f,\eps)$ of $P_1$, by Lemma \ref{lem:nfa3},
there is a path $y\stackrel{\sigma^R}{\Rightarrow}q$ in $N$ for some $\sigma$ and $q$, and
the path $m_0\stackrel{\rho^R}{\Rightarrow}x$ in $N$ splits into $m_0\stackrel{\rho_2^R}{\Rightarrow}z\stackrel{\rho_1^R}{\Rightarrow}x$
for some $\rho_2$, $z$ and $\rho_1\neq\eps$,
where $z\stackrel{\rho_1^R}{\Rightarrow}x$ does not start with an $\eps$-transition, and
there is a path $z\stackrel{\eps}{\rightarrow}w\stackrel{\upsilon^R}{\leadsto}s$ in $N$ for some $w$, $\upsilon$ and $s$,
where $q\stackrel{\sigma\rho_1/\upsilon}{\rightarrow}s$ is a transition in $P_1$
that is applied to the configuration $(q,\sigma\rho)$ in the run $(r,\tau\rho)\rightarrow^*(q_f,\eps)$.
Since $\rho_2$ is shorter than $\rho$, by induction, $z\stackrel{\eps}{\rightarrow}w$ is eventually in $E$.
During the iteration of the {\em backward} procedure in which $z\stackrel{\eps}{\rightarrow}w$ is added to $F$,
in view of the transition $q\stackrel{\sigma\rho_1/\upsilon}{\rightarrow}s$ in $P_1$ and the paths $w\stackrel{\upsilon^R}{\leadsto}s$ and
$z\stackrel{\rho_1^R}{\Rightarrow}x\stackrel{\eps}{\rightarrow}y\stackrel{\sigma^R}{\Rightarrow}q$ in $N$, the $\eps$-transitions in
the latter path are added to $E$. So in particular, $x\stackrel{\eps}{\rightarrow}y$ is added to $E$.
\qed
\end{proof}

\begin{theorem}
\label{thm:nfa2}
The returned set $U_2$ consists of the useless transitions in $P_1$.
\end{theorem}

\begin{proof}
Suppose the transition $\theta=q\stackrel{\sigma/\tau}{\rightarrow}r$ in $P_1$ is not in $U_2$. Since $\theta$ is in $P_1$,
by the {\em forward} procedure, there is a path $y\stackrel{\tau^R}{\leadsto}r$ in $N$. Since $\theta\not\in U_2$, while running
the {\em backward} procedure, for some $x$, the transition $x\stackrel{\eps}{\rightarrow}y$ was found to be in $E$,
and a path $x\stackrel{\sigma^R}{\Rightarrow}q$ was found to be in $N$.
In view of the transition $x\stackrel{\eps}{\rightarrow}y$ in $E$ and the path $y\stackrel{\tau^R}{\leadsto}r$ in $N$,
by Proposition \ref{prop:nfa2}, there is a path $m_0\stackrel{\rho^R}{\Rightarrow}x$ in $N$, for some $\rho$,
such that there is a run $(r,\tau\rho)\rightarrow^*(q_f,\eps)$ of $P_1$.
In view of the path $m_0\stackrel{\rho^R}{\Rightarrow}x\stackrel{\sigma^R}{\Rightarrow}q$ in $N$,
by Proposition \ref{prop:nfa1}, $(q,\sigma\rho)$ is reachable in $P_1$.
By applying $\theta$ to this configuration, $(r,\tau\rho)$ is reached. Since moreover
there is a run $(r,\tau\rho)\rightarrow^*(q_f,\eps)$ of $P_1$, $\theta$ is useful in $P_1$.

Vice versa, suppose $\theta$ is useful in $P_1$. Then there is a run
$(q_0,b_0)\rightarrow^*(q,\sigma\rho)\rightarrow(r,\tau\rho)\rightarrow^*(q_f,\eps)$ of $P_1$.
In view of the run $(q_0,b_0)\rightarrow^*(q,\sigma\rho)$,
by Proposition \ref{prop:nfa1}, there is a path $m_0\stackrel{\rho^R}{\Rightarrow}x\stackrel{\sigma^R}{\Rightarrow}q$ in $N$,
where we choose $x$ such that $x\stackrel{\sigma^R}{\Rightarrow}q$ does not start with an $\eps$-transition.
Since $\theta$ is in $P_1$, by the {\em forward} procedure, there is a path $x\stackrel{\eps}{\rightarrow}y\stackrel{\tau^R}{\leadsto}r$ in $N$.
In view of the run $(r,\tau\rho)\rightarrow^*(q_f,\eps)$, by Proposition \ref{prop:nfa3}, $x\stackrel{\eps}{\rightarrow}y$ is eventually in $E$.
In view of the paths $y\stackrel{\tau^R}{\leadsto}r$ and $x\stackrel{\sigma^R}{\Rightarrow}q$ in $N$, during the iteration
of the {\em backward} procedure in which $x\stackrel{\eps}{\rightarrow}y$ is added to $F$, $\theta$ is deleted from $U_2$.
\qed
\end{proof}

\paragraph{Complexity analysis}
Computing $U_2$ takes at most $O(Q^4T)$:
For each of the at most $O(Q^2)$ $\eps$-transitions in $E$, and for at most $T$
transitions in $P_1$, in step 3.3 a forward scan is performed over the $\eps$-transitions in $N$, which takes at most $O(Q^2)$.

\section{Implementation}

We made a prototype implementation of the algorithm, using a test suite of more than twenty pda's.
The largest pda, with 295 transitions, was obtained from the grammar of the programming language C.
This resulted in an nfa with 339 states and 1030 transitions, of which 695 $\eps$-transitions,
and took 11 seconds on a 2GHz processor.

Achieving this performance required two optimizations, both limiting
the influence of $\eps$-transitions.
The first concerns determining the set of states leading to $q$ in
step 2 of the {\em forward} procedure.
%% (Section \ref{sec:forward}).
%% this was optimized by collecting additional information about the
%% transition graph during the construction of the nfa.
This set is constructed by following paths backwards from $q$,
which may lead through webs of $\eps$-transitions, causing a considerable slow-down.
These $\eps$-transitions were created in step 5 of the {\em forward} procedure.
The optimization consists of computing for each state $s$, in step 5,
the set $B(s)$ of states that can reach $s$ through $\eps$-transitions only.
%% This mapping is used in subsequent iterations in step 3 (2?) to speed up path constructions.
If more $\eps$-transitions are added in step 5 during a next iteration, $B$ is updated.
The second optimization concerns memoization of $\eps$-transitions as they are encountered
on paths to $q$ in step 3.3 of the {\em backward} procedure.
%% (Section \ref{sec:backward});
%% this was optimized by memoization of intermediate results.
%% obtained at some states in the transition graph.
%% We record $\eps$-transitions as they are encountered on paths to $q$ in the nfa;
%% this result is thus immediately available when required in subsequent searches.

No further optimizations were applied. In fact, all sets were implemented as arrays;
choosing more advanced data structures would certainly improve the efficiency.

\bibliographystyle{plain}

%% \appendix
%% 
%% \section{A belated rebuttal}
%% 
%% This paper is a revised version of a submission to ICALP'2012. A reviewer of that submission
%% sketched an alternative solution to the problem of detecting useless transitions in pda's.
%% %% One of these solutions, which consists of marking a transition by a special input symbol $\xi$,
%% %% is briefly discussed in the introduction.
%% %DG%As a belated rebuttal we discuss here this supposed solution, only for the benefit of the reviewers of the current submission.
%% This rebuttal discusses that supposed solution, for the benefit of the reviewers of the current submission.
%% 
%% \begin{quote}
%% ``compute ${\it post}{\ast}(c_0)$ and ${\it pre}{\ast}(C_f)$ and then check, for any transition $pA \rightarrow qw$,
%% whether configurations starting with $pA$ and $qw$ are contained in those sets.''
%% \end{quote}
%% 
%% \noindent
%% A counter-example to the solution above is the pda consisting of the following four transitions:
%% $q_0\stackrel{\eps/a}{\rightarrow}q_1$, $q_0\stackrel{\eps/b}{\rightarrow}q_2$,
%% $q_1\stackrel{a/\eps}{\rightarrow}q_2$ and $q_2\stackrel{b/\eps}{\rightarrow}q_f$, with $F=\{q_f\}$.
%% The method above does not detect that the transition $q_1\stackrel{a/\eps}{\rightarrow}q_2$ is useless;
%% namely, ${\it post}{\ast}(q_0)$ contains a configuration starting with $q_1a$,
%% and ${\it pre}{\ast}(\{q_f\})$ contains a configuration starting with $q_2$.

\end{document}